\documentclass[twocolumn,english,amsthm]{autart}   

\usepackage{graphicx}          
                               
\usepackage{babel}
\usepackage{cite}
\usepackage{amsmath,amssymb}
\usepackage{booktabs}
\usepackage[shortcuts]{extdash}

\allowdisplaybreaks

\DeclareMathOperator*{\rms}{rms}
\DeclareMathOperator*{\argmin}{arg\;min}

\newtheorem{mydef}{Definition}
\newtheorem{assumption}{Assumption}
\newtheorem{remark}{Remark}

\begin{document}

\begin{frontmatter}

\title{Parametric identification of parallel Wiener-Hammerstein systems} 

\thanks[footnoteinfo]{The corresponding author is M.~Schoukens (maarten.schoukens@vub.ac.be).}
\thanks[johanPattyn]{The authors would like to thank Johan Pattyn for designing and building the parallel Wiener\-/Hammerstein system used in Section \ref{sec:measurement}.}
\thanks[preprint]{This paper is a postprint of a paper submitted to and accepted for publication in Automatica. This manuscript version is made available under the CC-BY-NC-ND 4.0 license. The published copy of record is available through: https://doi.org/10.1016/j.automatica.2014.10.105}
\author{Maarten Schoukens, Anna Marconato, Rik Pintelon, Gerd Vandersteen, Yves Rolain} 
\address{Vrije Universiteit Brussel (VUB), Dept. ELEC, Pleinlaan 2, B-1050 Brussels, Belgium}  


\begin{abstract}                          
	Block\-/oriented nonlinear models are popular in nonlinear modeling because of their advantages to be quite simple to understand and easy to use. To increase the flexibility of single branch block\-/oriented models, such as Hammerstein, Wiener, and Wiener\-/Hammerstein models, parallel block\-/oriented models can be considered. This paper presents a method to identify parallel Wiener\-/Hammerstein systems starting from input-output data only. In the first step, the best linear approximation is estimated for different input excitation levels. In the second step, the dynamics are decomposed over a number of parallel orthogonal branches. Next, the dynamics of each branch are partitioned into a linear time invariant subsystem at the input and a linear time invariant subsystem at the output. This is repeated for each branch of the model. The static nonlinear part of the model is also estimated during this step. The consistency of the proposed initialization procedure is proven. The method is validated on real-world measurements using a custom built parallel Wiener\-/Hammerstein test system.
\end{abstract}

\end{frontmatter}

\section{Introduction} \label{sec:introduction}
	Nonlinear models are much needed these days to improve plant control performance, to gain better insight in the behavior of the system under test, or to compensate for a potential nonlinear behavior. Due to the separation of the nonlinear dynamic behavior into linear time invariant (LTI) dynamics and the static nonlinearities (SNL), block\-/oriented nonlinear models are quite simple to understand and easy to use. 
	
	A wide variety of block\-/oriented models has been studied over the last years including Hammerstein (Nonlinear static - Linear dynamic or N-L connection) and Wiener models (L-N) \cite{Giri2010}. This type of single branch models can be extended to Hammerstein-Wiener models (N-L-N) \cite{Bai1998,Crama2004b,SchoukensM2012}, or Wiener\-/Hammerstein models (L-N-L) \cite{Billings1978,Vandersteen1997,Sjoberg2012,Westwick2012,SchoukensM2014}. To increase the flexibility of the single branch block\-/oriented models even more, parallel block\-/oriented models can be considered such as parallel Hammerstein \cite{Gallman1975,SchoukensM2011}, and parallel Wiener models \cite{SchoukensM2012b,Lyzell2012a,SchoukensM2013a}. 
	
	This paper presents a method to identify parallel Wiener\-/Hammerstein systems, whose structure is shown in Figure \ref{fig:parallelWienerHammerstein}. Previously published methods \cite{Baumgartner1975,Wysocki1976,Billings1979} studied a subclass of the parallel Wiener\-/Hammerstein structure that is called the $S_M$ model structure. Identification methods based on repeated sine measurements \cite{Baumgartner1975,Wysocki1976}, or white Gaussian inputs \cite{Billings1979} are available for this model structure. In \cite{Palm1978,Palm1979} it is shown that a wide class of Volterra systems can be approximated arbitrary well using a parallel Wiener\-/Hammerstein model structure. However, no method is presented there to identify such models.
	
	The $S_M$ identification method presented in \cite{Billings1979} uses Gaussian excitation signals, like the method presented in this paper. However, the $S_M$ method is a generalization of a Wiener-Hammerstein identification algorithm based on a parametrized version of higher order correlation functions between input and output \cite{Billings1978}. This approach has been compared in \cite{SchoukensM2014} with two other approaches \cite{Westwick2012,SchoukensM2014}, and it was outperformed by these alternatives. The main problem of the method seems to be the noise sensitivity.
	
	The parallel Wiener\-/Hammerstein identification approach proposed here combines the parallel Hammerstein and parallel Wiener identification methods presented in \cite{SchoukensM2011,SchoukensM2012b} with a specific initialization approach for Wiener\-/Hammerstein systems presented in \cite{Sjoberg2012}. This paper hereby extends the results of \cite{SchoukensM2013b}. In the paper presented here, the consistency of the proposed initialization procedure is proven, the computational aspects of the proposed method are discussed, the positive effect of the initialization method is shown, and the method is applied to a real-world measurement example.
	
	The outline of the paper is as follows. Section \ref{sec:classIntro} introduces the system and signal classes, and the stochastic framework used. Section \ref{sec:identifiability} discusses the identifiability of a parallel Wiener\-/Hammerstein system. Next, the best linear approximation (BLA) of a parallel Wiener\-/Hammerstein system is studied in Section \ref{sec:bla}. The identification algorithm for parallel Wiener\-/Hammerstein systems is explained in Section \ref{sec:estimation}. Section \ref{sec:persistence} discusses the persistence of excitation, Section \ref{sec:consistency} proves the consistency of the proposed identification method. A final, jointly nonlinear least squares optimization with respect to all the parameters of all the blocks is performed in Section \ref{sec:optimization}. Some computational aspects of the method are discussed in Section \ref{sec:computationAspects}. Finally, the good performance of the proposed method is illustrated in Section \ref{sec:measurement} on real-world measurements using a custom built parallel Wiener\-/Hammerstein test system. The positive effect of the proposed initialization method on the performance of the optimized model is also shown in this section.
	
	\begin{figure}
		\centering
			\includegraphics[width=0.95\columnwidth]{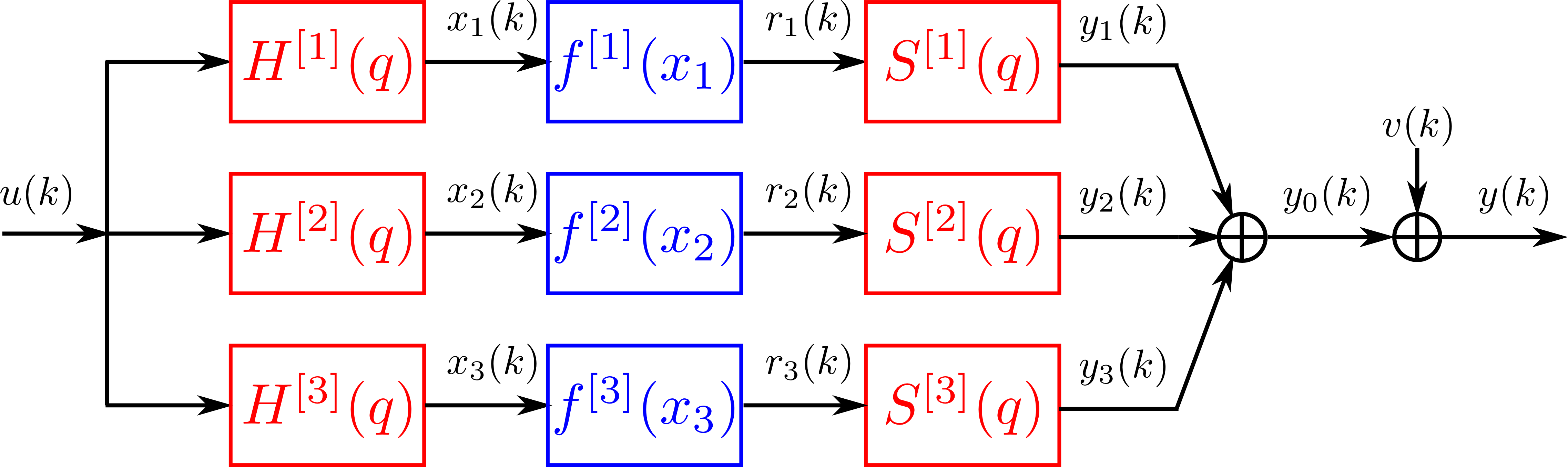}
		\caption{A 3-branch parallel Wiener\-/Hammerstein system: a parallel connection of Wiener\-/Hammerstein systems. The static nonlinear block $f^{[i]}$ of the $i$-th branch is sandwiched in between the LTI blocks $H^{[i]}(q)$ and $S^{[i]}(q)$. The noise source $v(k)$ is additive colored noise.}
		\label{fig:parallelWienerHammerstein}
	\end{figure}	
		
\section{System, signals and stochastic framework} \label{sec:classIntro}
		This section describes the system and signal classes, and introduces the stochastic framework considered in this paper.

		\begin{mydef}
			\textbf{Riemann equivalence class of asymptotically normally distributed excitation signals.}
			Consider a signal $u$ with a power spectrum $S_U(j\omega)$, which is piecewise continuous, with a finite number of discontinuities. A random signal belongs to the Riemann equivalence class of $u$ if it obeys by any
of the following statements:
			\begin{enumerate}
				\item It is a Gaussian noise excitation with power spectrum $S_U(j\omega)$.
				\item It is a random multisine or random phase multisine \cite{Pintelon2012} such that:
					\begin{align}
						\frac{1}{N} \sum_{k = k_1}^{k_2} E\left\{ \left| U\left( j\omega_k \right) \right|^2 \right\} &= \frac{1}{2 \pi} \int_{\omega_{k_1}}^{\omega_{k_2}} S_{U}\left( \nu \right) d\nu  + O\left( N^{-1} \right), \nonumber
					\end{align}
					with $\omega_k = k \frac{2 \pi f_s}{N}$, $k \in \mathbb{N}$, $0<\omega_{k_1}<\omega_{k_2}<\pi f_s$, and $f_s$ the sample frequency.
			\end{enumerate}		
		\end{mydef}
		
		\begin{assumption} \label{ass:input}
			The excitation signal $u(k)$ is stationary and belongs to the Riemann equivalence class of asymptotically normally distributed excitation signals.
		\end{assumption}

		\begin{assumption} \label{ass:noise}
			An additive, colored zero-mean noise source $v(k)$ with a finite variance is present at the output of the system only:
				\begin{align}
					y(k) = y_0(k) + v(k), \label{eq:noise}
				\end{align}
				where $y(k)$, $y_0(k)$ and $v(k)$ are scalars. The noise $v(k)$ is assumed to be independent from the known input $u(k)$.
		\end{assumption}
		Assumption \ref{ass:noise} excludes that the system operates in closed loop.

		The class of parallel Wiener\-/Hammerstein systems is considered. A parallel Wiener\-/Hammerstein system consists of a parallel connection of different Wiener\-/Hammerstein systems that share the same input signal. The output of the total system is obtained as the sum of the outputs of the different branches. A parallel Wiener\-/Hammerstein system with three parallel branches is shown in Figure \ref{fig:parallelWienerHammerstein}.
		
		The noiseless output $y_0(k)$ of a parallel Wiener\-/Hammerstein system is given by:
		\begin{align}
			y_0(k) &= \sum_{i=1}^{n_{br}} y_i(k), \label{eq:ParWhOuptut1} \\
			y_i(k) &= S^{[i]}(q) r_i(k), \\
			r_i(k) &= f^{[i]}(x_i(k)), \\
			x_i(k) &= H^{[i]}(q) u(k), \label{eq:ParWhOuptut2} 
		\end{align}
		where $n_{br}$ is the number of parallel branches in the parallel Wiener\-/Hammerstein system, $H^{[i]}(q)$ and $S^{[i]}(q)$ are the front and back discrete time representations of the LTI blocks present in branch $i$, $f^{[i]}(x_i(k))$ is the static nonlinear block present in branch $i$, and the signals are as shown in Figure \ref{fig:parallelWienerHammerstein}.
		
		All the LTI blocks are considered to be modeled by stable infinite impulse response (IIR) filters, parameterized by a rational function in the backwards shift operator $q^{-1}$:
		\begin{align}
			H^{[i]}(q) &= \frac{B_{h}^{[i]}(q)}{A_{h}^{[i]}(q)}, \label{eq:Lti1} \\
					 	 &= \frac{b_{h,0}^{[i]} + b_{h,1}^{[i]}q^{-1} + \ldots + b_{h,n_{b_{h},i}}^{[i]}q^{-n_{b_{h},i}}}{a_{h,0}^{[i]} + a_{h,1}^{[i]}q^{-1} + \ldots + a_{h,n_{a_{h},i}}^{[i]}q^{-n_{a_{h},i}}}, \nonumber  \\
			S^{[i]}(q) &= \frac{B_{s}^{[i]}(q)}{A_{s}^{[i]}(q)}, \label{eq:Lti2} \\
					 	 &= \frac{b_{s,0}^{[i]} + b_{s,1}^{[i]}q^{-1} + \ldots + b_{s,n_{b_{s},i}}^{[i]}q^{-n_{b_{s},i}}}{a_{s,0}^{[i]} + a_{s,1}^{[i]}q^{-1} + \ldots + a_{s,n_{a_{s},i}}^{[i]}q^{-n_{a_{s},i}}}, \nonumber 
		\end{align}
		where $n_{b_{h},i}$ and $n_{a_{h},i}$ are respectively the finite orders of the numerator and denominator of the front dynamics of the $i$-th parallel branch, and $n_{b_{s},i}$ and $n_{a_{s},i}$ are the orders of the numerator and denominator of the back dynamics of the $i$-th parallel branch.
		
		The static nonlinear function $f^{[i]}(x_i(k))$ contained in the $i$th branch is described by a linear combination of $n_f$ nonlinear basis functions:
		\begin{align}
			f^{[i]}(x_i(k)) &= \sum_{j=1}^{n_f} \beta_f^{[i]} f_j^{[i]}(x_i(k)). \label{eq:snl}
		\end{align}
		Each basis function $f_j^{[i]}(x)$ is assumed to have a finite output for any finite input $x$.	Examples of such nonlinear functions are polynomial functions, piecewise linear functions or radial basis function networks.
		
		\begin{assumption} \label{ass:system} 
			The true system is a discrete time parallel Wiener\-/Hammerstein system, as described by eq. (\ref{eq:noise}) to (\ref{eq:snl}).
		\end{assumption}
		
		The parallel Wiener\-/Hammerstein system class that is used here is a more general system class than the $S_M$ system class that is used in \cite{Baumgartner1975,Wysocki1976,Billings1979}. The $S_M$ model has $M$ parallel branches, and the $m$-th branch contains a monomial nonlinearity equal to $(.)^m$. This restricts the model to have a polynomial nonlinearity only, and to contain only one branch for each degree of this polynomial nonlinearity. Thus a parallel Wiener-Hammerstein system containing two parallel branches, each with different LTI subsystems, and with different polynomial nonlinearities can, in general, not be modeled by a $S_M$ model. The method that is presented in this paper also makes some extra assumptions on the parallel Wiener-Hammerstein system in the following sections. However, even when these assumptions are met, the considered system class still allows for a much more complicated nonlinear system behavior.		
		
\section{Identifiability} \label{sec:identifiability}
	The problem of identifying a parallel Wiener\-/Hammerstein system inherits all the identifiability issues that are present in the identification of a Wiener\-/Hammerstein system \cite{SchoukensM2014,SchoukensM2013b}: a gain exchange between the LTI blocks and the static nonlinear block leads to a degeneracy in the parameter space. There can also be a delay exchange between the front LTI blocks and the back LTI blocks, but only a finite number of delay exchanges values are possible when a parametric transfer function model is used to model the LTI blocks. A degeneration in the parameter space results in multiple parameterizations that lead to the same input-output behavior of the system. The rank of the Jacobian matrix of the model is reduced by one for each degeneration.
	
	An additional identifiability issue appears due to the parallel nature of the parallel Hammerstein, the parallel Wiener, and the parallel Wiener\-/Hammerstein systems \cite{SchoukensM2012b,SchoukensM2013b}. Starting from input-output data only, infinitely many equivalent models can be obtained by linear transformation of one of the models. This introduces a full rank linear transformation between the outputs of the front dynamic blocks $H^{[i]}(q)$ and the inputs of the static nonlinearities of the different branches. A similar full rank linear transformation can be introduced between the outputs of the static nonlinearities and the inputs of the back LTI blocks $S^{[i]}(q)$. Such a full rank transformation results in a model structure that differs from the model structure presented in Figure \ref{fig:parallelWienerHammerstein}. The full rank linear transformations that can be present between the front LTI blocks and the static nonlinear blocks, and between the static nonlinear blocks and the back LTI blocks, can be incorporated in the static nonlinear blocks. This transforms the SISO static nonlinearities of each branch into one MIMO static nonlinearity, as is shown in Figure \ref{fig:parallelWienerHammersteinModel}.
	
	The number of degenerations $n_{deg}$ present in the model is quantified by:
	\begin{align}
		n_{deg} = 2n_{br}^2,
	\end{align}
	where $n_{br}$ is the number of parallel branches in the model. Each full rank linear transform (which includes also the gain exchanges) introduces $n_{br}^2$ degenerations in the model.
	
	The model is intended to describe the system, and has to overcome all of the identifiability issues. The gain and delay exchanges can be accounted for by using an appropriate normalization and parameterization \cite{SchoukensM2014}. The full rank linear transformations, on the other hand, require some attention.  As a consequence of the full rank linear transformations, the model with one SISO static nonlinearity for each branch is transformed into a model with one MIMO static nonlinearity that describes the nonlinear behavior of the system. This modified model structure is shown in Figure \ref{fig:parallelWienerHammersteinModel}. In a later step, the MIMO static nonlinearity can be decoupled again to yield one SISO static nonlinearity for each branch \cite{Tiels2013,SchoukensM2014a} hereby eliminating cross-coupling between branches. The identified LTI blocks will be a linear combination of the amplitude scaled and/or delayed versions of the exact but unknown LTI blocks that are present in the system.
	
	\begin{figure}
		\centering
			\includegraphics[width=0.95\columnwidth]{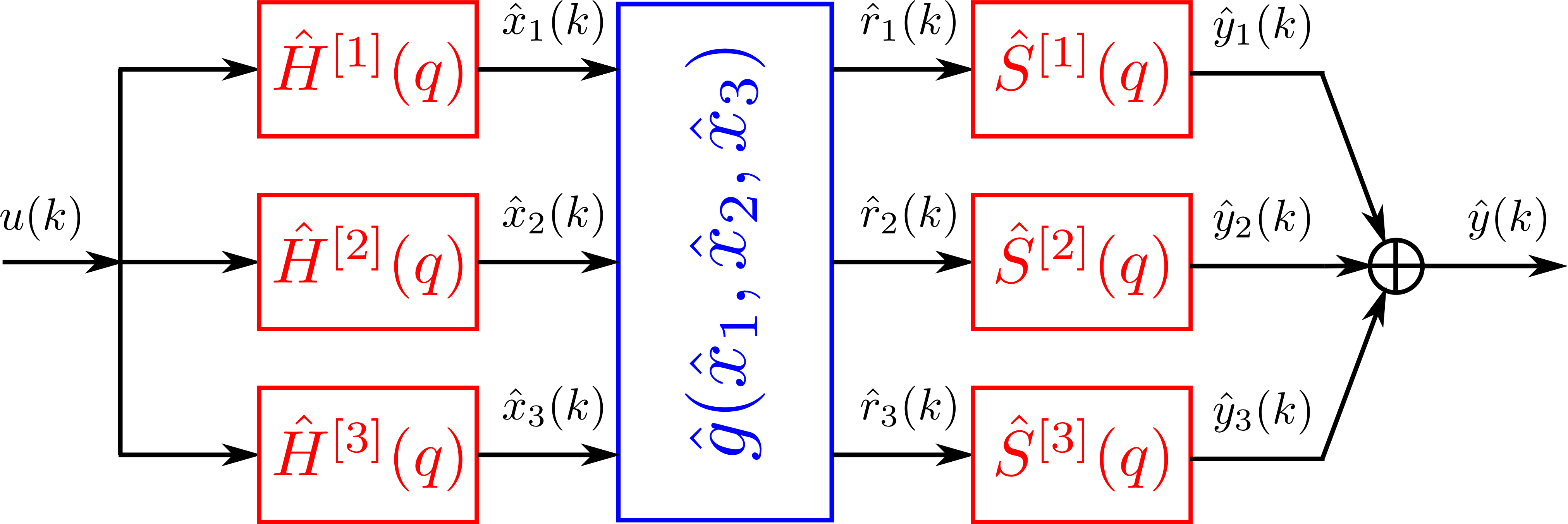}
		\caption{A 3-branch parallel Wiener\-/Hammerstein model. A MIMO static nonlinear block $\hat{g}$ is sandwiched in between the LTI blocks $\hat{H}^{[i]}(q)$ and $\hat{S}^{[i]}(q)$.}
		\label{fig:parallelWienerHammersteinModel}
	\end{figure}	
	
\section{The best linear approximation} \label{sec:bla}

The best linear approximation (BLA) approximates the output of a nonlinear system with the response of an LTI model in mean square sense. The BLA depends on the system, on the probability density function of the chosen input signal, and on the input power spectrum (rms value and coloring).

\begin{mydef} \textbf{Best linear approximation (BLA).}
	The BLA of a nonlinear system is the linear system $G_{bla}(q)$ that minimizes the mean square error \cite{Enqvist2005a,Enqvist2010,Pintelon2012}:
	\begin{align}
		G_{bla}(q) &=  \underset{G(q)}{\argmin} \; E\left\{ \left( \tilde{y}(k) - G(q)\tilde{u}(k) \right)^2 \right\}, \nonumber \\
		\tilde{u}(k) &= u(k) - E\left\{u(k)\right\}, \\
		\tilde{y}(k) &= y(k) - E\left\{y(k)\right\}, \nonumber
	\end{align}
	where the expectation $E\left\{.\right\}$ is taken with respect to the random realization of $u(k)$.
\end{mydef}

The BLA of a parallel Wiener\-/Hammerstein system obtained under Assumption \ref{ass:input} is a simple function of the dynamic blocks that are present in the parallel Wiener\-/Hammerstein system under test \cite{Bussgang1952,Enqvist2010,Pintelon2012}. The static nonlinearity $f^{[i]}(x_i)$ of branch $i$ in a parallel Wiener\-/Hammerstein system can be approximated by a constant gain $\alpha^{[i]}$ \cite{Bussgang1952,Enqvist2010}. This results in Theorem \ref{theo:BLA}.
		
		\begin{thm} \label{theo:BLA}
			The BLA of a parallel Wiener\-/Hammerstein system (Assumption \ref{ass:system}) excited by inputs satisfying Assumption \ref{ass:input} is given by:
			\begin{align}
				G_{bla}(q) = \sum_{i=1}^{n_{br}} \alpha^{[i]} H^{[i]}(q) S^{[i]}(q), \label{eq:BLA}
			\end{align}	
			where $\alpha^{[i]}$ depends on the subsystems in the $i$th branch, the power spectrum of the input signal $u$, and hence as well on the variance of the input signal $u$.
		\end{thm}
		\begin{proof}
			It is shown in Section 3.4.3.5 of \cite{Pintelon2012} that the BLA of the i-th branch of a parallel Wiener\-/Hammerstein system is given by $\alpha^{[i]} H^{[i]}(q) S^{[i]}(q)$. Since the output of a parallel Wiener\-/Hammerstein system is given by the sum of the different Wiener\-/Hammerstein branches, the BLA of a parallel Wiener\-/Hammerstein system is given by eq. (\ref{eq:BLA}).
		\end{proof}

		\begin{assumption} \label{ass:evenNl}
			The BLA $\alpha^{[i]} H^{[i]}(q) S^{[i]}(q)$ of every branch $i$ has a non-zero gain $\alpha^{[i]}$. 
		\end{assumption}
		It can happen that the BLA of one of the branches of the parallel Wiener-Hammerstein system is equal to zero, or in other words $\alpha^{[i]} = 0$. This is the case when the nonlinear function $f^{[i]}(x_{i})$ is even around the expected value of $x_{i}$.	In this case, a BLA of a reduced order is obtained that does not contain the dynamics of branch $i$. This assumption excludes that the static nonlinearity $f^{[i]}(x_i)$ of branch $i$ is symmetric with respect to the DC setpoint of the signal $x_i$. Different DC setpoints can be tried to avoid a zero gain $\alpha^{[i]}$.
		
		\begin{assumption} \label{ass:polezero}
			The combined dynamics $H^{[i]}(q)S^{[i]}(q)$ of the $i$th branch do not contain any pole\-/zero cancellation for any branch $i$. 
		\end{assumption}
		A pole that appears in the front LTI block of a branch, can be canceled by a zero that is present in the back LTI block of the same branch. They will not be detected during the parametrization of the BLA. This assumption is quite common for Wiener-Hammerstein identification algorithms, see for example the two identification algorithms presented in \cite{Sjoberg2012a}. However, there exist different Wiener-Hammerstein and $S_M$ identification algorithms that do not need this assumption \cite{Baumgartner1975,Wysocki1976,Billings1979,SchoukensM2014}. This is possible using a more advanced correlation analysis \cite{Billings1979,SchoukensM2014}, or a more restrictive class of input signals \cite{Baumgartner1975,Wysocki1976}.
		
		An important observation with respect to eq. (\ref{eq:BLA}) is that the input dependent gain $\alpha^{[i]}$ only appears in the numerator: 
		\begin{align} \label{eq:BlaCommonDen}
			G_{bla}(q) = \frac{\sum_{i=1}^{n_{br}}\alpha^{[i]}B_{hs}^{[i]}(q)\prod_{j=1,j\neq i}^{n_{br}}A_{hs}^{[j]}(q)} {\prod_{i=1}^{n_{br}}A_{hs}^{[i]}(q)}, 
		\end{align}
		where
		\begin{align}
			B_{hs}^{[i]}(q) &= B_{h}^{[i]}(q)B_{s}^{[i]}(q), \\
			A_{hs}^{[i]}(q) &= A_{h}^{[i]}(q)A_{s}^{[i]}(q).
		\end{align}
		This means that under Assumptions \ref{ass:evenNl} and \ref{ass:polezero}, the poles of the identified BLA are also the poles of the LTI blocks that are present in the system. The zeros of the BLA of a parallel Wiener\-/Hammerstein system may change when the amplitude, power spectrum, or the offset (DC value) of the input signal changes.

\section{Estimating the parameters of a parallel Wiener\-/Hammerstein system} \label{sec:estimation}
	The approach presented in \cite{SchoukensM2011,SchoukensM2012b} to decompose the dynamics over the different branches of a parallel Hammerstein and parallel Wiener systems is combined with an initialization approach that splits the dynamics into the front and back LTI blocks of a Wiener\-/Hammerstein system as presented in \cite{Sjoberg2012}. 
	
	Other approaches to split the dynamics of a Wiener\-/Hammerstein system exist in the literature \cite{Billings1978,Westwick2012,SchoukensM2014}, but are more complex to implement and seem to be more sensitive to noisy data. An initial version of this method was presented in \cite{SchoukensM2013b}. 

	The proposed approach starts with an estimation of the BLA of the considered system for different operating conditions (Section \ref{sec:measBla}). The different operating conditions are obtained using input signals with different power spectra. This includes the use of different magnitudes, different offsets, or different coloring of the power spectra. A consistent estimate of the overall dynamics that are present in the nonlinear parallel Wiener\-/Hammerstein system results. 
	
	Next, the measured BLAs are parameterized (Section \ref{sec:paramBla}) using a different LTI model for each operating condition. A common denominator model is used for all operating conditions simultaneously. This is indeed possible, as Theorem \ref{theo:BLA} assures that the poles of the different measured BLAs are the same. 
	
	Starting from the parameterized BLAs, a decomposition of the overall dynamics at the different operating conditions is calculated in Section \ref{sec:decomposing}. It uses the singular value decomposition (SVD) of a matrix constructed using the numerator coefficients of the parameterized BLAs obtained at the different operating conditions. This step results in an estimate of the number of branches that is present in the parallel Wiener\-/Hammerstein system. The number is obtained based on the estimated rank of the decomposed matrix. The dynamics $H^{[i]}(q)S^{[i]}(q)$ that are present in each branch are estimated next, up to the identifiability issues presented in Section \ref{sec:identifiability}. 
	
	Finally, a modified version of the algorithm proposed in \cite{Sjoberg2012} is proposed in Section \ref{sec:partitioning} to partition the dynamics $H^{[i]}(q)S^{[i]}(q)$ over the different blocks of the parallel Wiener\-/Hammerstein model, and to estimate the static nonlinearity that is present in the model.
	
	\subsection{Estimating a BLA for different operating conditions} \label{sec:measBla}
		The nonparametric frequency response function (FRF) estimate of the BLA at operating condition $i_r$ is labeled $\hat{G}_{bla}^{[i_r]}(j\omega_k)$. It is obtained by the robust BLA estimation method proposed in \cite{Pintelon2012, Schoukens2012}. Both the FRF and the sample variance $\hat{\sigma}^2_{\hat{G}_{bla}^{[i_r]}}(j\omega_k)$ of the BLA are obtained at each excited frequency. The latter is used to determine the weighting factor used during the parameterization of the BLA. This process is explained in more detail in \cite{Pintelon2012, Schoukens2012, SchoukensM2013b}.
			
	\subsection{Parameterizing the BLAs} \label{sec:paramBla}
		The measured nonparametric BLAs $\hat{G}_{bla}^{[i_r]}(j\omega_k)$ at the $R$ different operating conditions are parameterized simultaneously using a common denominator model. To perform the estimation, a weighted total least squares initialization is used \cite{Pintelon1998}. It is followed by a sample maximum likelihood estimation \cite{Pintelon2011b,Pintelon2012}. The frequency dependent estimation weights for the FRF are inversely proportional to the estimated sample variances $\hat{\sigma}^2_{\hat{G}_{bla}^{[i_r]}}(j\omega_k)$ of the BLAs for the $R$ different operating conditions of the system. This results in a parameterized version of the different BLAs $\hat{G}_{bla}\left(q,\hat{\boldsymbol{\theta}}_{bla}\right)$:
		\begin{align}
			\hat{G}_{bla}\left(q,\hat{\boldsymbol{\theta^{[i_r]}}}_{bla}\right) &= \frac{\hat{d}_{0}^{[i_r]} + \hat{d}_{1}^{[i_r]}q^{-1} + \ldots + \hat{d}_{n_{d}}^{[i_r]}q^{-n_{d}}}{\hat{c}_{0} + \hat{c}_{1}q^{-1} + \ldots + \hat{c}_{n_{c}}q^{-n_{c}}}, \label{eq:parBla}
		\end{align}
		where the denominator coefficients are shared by the BLAs for the different operating conditions $i_r$, while the numerator coefficients vary with the input operating condition $i_r$. $\hat{\boldsymbol{\theta}}_{bla}$ contains all the denominator coefficients $\hat{c}_i$, and all the numerator coefficients $\hat{d}_i^{[i_r]}$ of the BLAs for the different operating conditions $i_r$. The model order of the parametrized BLAs can be selected using standard model structure selection techniques \cite{Ljung1999}.		
		
		\begin{remark} 
			Assumption \ref{ass:system} considers discrete time linear time-invariant systems. The proposed method works equally well for continuous time systems, estimating a continuous time model. Instead of parameterizing the LTI blocks using a rational function of finite order in the backwards shift operator $q^{-1}$, a continuous time $s$-domain parameterization can be used.
		\end{remark}			

	\subsection{Decomposing the BLAs} \label{sec:decomposing}
		The overall frequency dynamics need to be distributed over the different LTI systems that are present in the branches at the front and the back of the parallel Wiener\-/Hammerstein model. This section presents a decomposition of the numerator coefficients of the overall dynamics of the BLA into a set of basis vectors that describe the space spanned by the numerator vectors. These basis vectors are an estimate of the dynamics of each parallel branch.
		
		A difference with the previous approaches in \cite{SchoukensM2011,SchoukensM2012b} is that the numerators of the estimated BLAs are decomposed, rather than the nonparametric BLA transfer functions. This is possible and adequate since a common denominator model is used for the parameterized BLAs. This new method avoids a frequency sampling step of the parametric BLAs, and a re-parameterization of the decomposed BLA dynamics. The process is explained next.
			
		First, a matrix $\hat{\boldsymbol{D}}$ is constructed containing the stacked estimated numerator coefficients of the BLAs at the different operating conditions:
		\begin{align}
			\hat{\boldsymbol{D}} &= \left[ \begin{array}{cccc}
											\hat{d}_{0}^{[1]} & \hat{d}_{1}^{[1]} & \ldots & \hat{d}_{n_{d}}^{[1]} \\
											\hat{d}_{0}^{[2]} & \hat{d}_{1}^{[2]} & \ldots & \hat{d}_{n_{d}}^{[2]} \\
											\vdots						& \vdots						& \ddots & \vdots								 \\
											\hat{d}_{0}^{[R]} & \hat{d}_{1}^{[R]} & \ldots & \hat{d}_{n_{d}}^{[R]} \\
					 					\end{array} \right]. \label{eq:D}
		\end{align}
		The underlying distortion free version of this matrix, $\boldsymbol{D}$, is of low rank. The maximum rank of the matrix $\boldsymbol{D}$, for $R, \; n_d > n_{br}$ is $n_{br}$. Using eq. \eqref{eq:BlaCommonDen}, one can write the $\boldsymbol{D}$ matrix as:
		\begin{align}
			\boldsymbol{D} &= \boldsymbol{A}\boldsymbol{B}, \\
			\boldsymbol{A} &= \left[ \begin{array}{cccc}
											\alpha_{1}^{[1]} 			& \alpha_{1}^{[2]} & \ldots & \alpha_{1}^{[n_{br}]} \\
											\alpha_{2}^{[1]} 			& \alpha_{2}^{[2]} & \ldots & \alpha_{2}^{[n_{br}]} \\
											\vdots					 			& \vdots					 & \ddots & \vdots								\\
											\alpha_{R}^{[1]} 			& \alpha_{R}^{[2]} & \ldots & \alpha_{R}^{[n_{br}]} \\
					 					\end{array} \right], \\
			\boldsymbol{B} &= \left[ \begin{array}{cccc}
											b_{0}^{[1]} 		 & b_{1}^{[1]} 			& \ldots & b_{n_d}^{[1]} 			\\
											b_{0}^{[2]} 		 & b_{1}^{[2]} 			& \ldots & b_{n_d}^{[2]} 			\\
											\vdots					 & \vdots						& \ddots & \vdots							\\
											b_{0}^{[n_{br}]} & b_{1}^{[n_{br}]} & \ldots & b_{n_d}^{[n_{br}]} \\
					 					\end{array} \right],
		\end{align}
		where $\alpha_{i}^{[j]}$ is the gain of the $j$-th branch of the $i$-th BLA, and $b_{i}^{[j]}$ is the $i$-th degree coefficient of $B_{hs}^{[i]}(q)\prod_{j=1,j\neq i}^{n_{br}}A_{hs}^{[j]}(q)$ (see eq. \eqref{eq:BlaCommonDen}). 
		
		The matrix $\boldsymbol{B}$ depends only on the dynamics that are present in the different branches of the system. The matrix $\boldsymbol{A}$ depends both on the system and on the input signal.
		
		\begin{assumption} \label{ass:D}
			The rank of the BLA numerator matrix $\boldsymbol{D}$ is equal to the number of parallel branches in the system.
		\end{assumption}
		The proposed identification method is based on a decomposition of the BLA over the different branches of the parallel Wiener\-/Hammerstein model. For the method to work, this decomposition should be able to separate the dynamics of each branch. This implies that the numerator of the combined dynamics of one branch ($H^{[i]}(q)S^{[i]}(q)$) of one branch is linearly independent from the numerators of the combined dynamics of the other branches of the parallel Wiener\-/Hammerstein system. This assumption also excludes the particular case of a parallel Wiener-Hammerstein system that consists of two LTI or two static nonlinear blocks placed in parallel, or for example a parallel Wiener-Hammerstein system where $H^{[1]}(q) = S^{[2]}(q)$ and $H^{[2]}(q) = S^{[1]}(q)$.
		
		The SVD of $\hat{\boldsymbol{D}}$ yields an orthonormal basis for the space spanned by the $\hat{\boldsymbol{D}}$-matrix:
		\begin{align}
			\hat{\boldsymbol{D}} &= \boldsymbol{U}_{bla} \boldsymbol{\Sigma}_{bla} \boldsymbol{V}_{bla}^T,
		\end{align}
		where superscript $.^T$ denotes the transpose of a matrix, $\boldsymbol{V}_{bla}$ contains the right singular vectors which act as an orthonormal basis for the right hand side space $\hat{\boldsymbol{D}}$-matrix, $\boldsymbol{\Sigma}_{bla}$ is a diagonal matrix containing the singular values, and $\boldsymbol{U}_{bla}$ contains the basis for the left hand side space.
		
		The column vectors in $\boldsymbol{V}_{bla}$ provide an estimate of the numerator coefficients for each branch:
		\begin{align}
			\hat{G}_{i_{br}}(q) &= \frac{\hat{\delta}_{0}^{[i_{br}]} + \hat{\delta}_{1}^{[i_{br}]}q^{-1} + \ldots + \hat{\delta}_{n_{d}}^{[i_{br}]}q^{-n_{d}}}{\hat{c}_{0} + \hat{c}_{1}q^{-1} + \ldots + \hat{c}_{n_{c}}q^{-n_{c}}},
		\end{align}
		where $\hat{\delta}_{j}^{[i_{br}]}$ is the element of the $j$-th row and $i_{br}$-th column of the matrix $\boldsymbol{V}_{bla}$. 
		
		The rank of the matrix $\boldsymbol{D}$ corresponds to the number of parallel branches $n_{br}$ that is necessary to describe the system. This rank can be obtained by applying a rank estimation algorithm on the singular value matrix $\boldsymbol{\Sigma}_{bla}$ \cite{Rolain1997}, that is obtained from the noisy matrix $\hat{\boldsymbol{D}}$. To do so, the column covariance matrix $\boldsymbol{C}_{D}$ of $\hat{\boldsymbol{D}}$ is needed. This column covariance matrix is obtained from the covariance of the parameters estimated in the BLA parametrization step. The whitened matrix ${\boldsymbol{D}}_{white}$ is given by:
		\begin{align}
			\boldsymbol{D}_{white} = \hat{\boldsymbol{D}}\boldsymbol{C}_{D}^{-1/2}
		\end{align}
		The estimated rank of the noisy matrix $\hat{\boldsymbol{D}}$ corresponds to the number of singular values of $\boldsymbol{D}_{white}$ that are higher than 1 \cite{Rolain1997}. The reader is referred to \cite{Rolain1997} for more details about the rank estimation method and its hypotheses.
		
	\subsection{Partitioning the poles and zeros} \label{sec:partitioning}
		This section presents an algorithm to partition the dynamics of each branch $\hat{G}_{i_{br}}(q)$ over the front and the back dynamics. The basic idea is pretty simple: try every partition of poles and zeros in the different LTI blocks, estimate the static nonlinear block with a fixed set of nonlinear basis functions, and finally select the model that minimizes the simulation error.
	
		\subsubsection{Generating all pole and zero partitions} \label{sec:pattern}
		
			\begin{assumption} \label{ass:poles}
				The front dynamic block of branch $i$ ($i = 1 \ldots n_{br}$) and the back dynamic block of branch $j$ ($j = 1 \ldots n_{br}$) have no common poles, wherever $i \neq j$.
			\end{assumption}
			This Assumption allows one to assign each estimated BLA pole to either the front or the back dynamics. A pole that is present in two different branches only appears once in the BLA. This does not pose a problem, if that pole is originating from either the front or the back LTI blocks due to the common denominator approach. However, this creates a problem when that pole is present once in the front LTI block of one branch and once in the back LTI block of another branch since it can only be assigned to either the front or the back LTI blocks.
		
			A first step in the algorithm is to generate all possible pole and zero partitions for the different LTI blocks. The poles and zeros to be distributed are the ones obtained from the branch dynamic estimated before. Let $\hat{G}_{i_{br}}(q)$ be the dynamics of branch $i_{br}$ of the parallel Wiener\-/Hammerstein model. Under Assumption \ref{ass:poles}, every pole and zero of $\hat{G}_{i_{br}}(q)$ has to be assigned to either the front or the back LTI block of the $i_{br}$-th branch. Some of the computational aspects of this approach are discussed in Section \ref{sec:computationAspects}. Complex pole and/or zero pairs are allocated pairwise to impose real coefficients in the transfer function model. The common denominator approach is preserved during the partitioning procedure. The construction of the front and the back dynamic systems of the branch $i_{br}$ is then:	
			\begin{align}
				\hat{G}_{i_{br}}(q) &= \gamma_{i_{br}} \frac{\hat{B}^{\{z_j^{i_{br}}\}}_{h}(q)}{\hat{A}^{\{p_i\}}_{h}(q)} \frac{\hat{B}^{\{z_j^{i_{br}}\}}_{s}(q)}{\hat{A}^{\{p_i\}}_{s}(q)} \label{eq:dynSplit}
			\end{align}
			for all possible pole partitions $\{p_i\}$, and for all possible zero partitions $\{z_{j}^{i_{br}}\}$ of branch $i_{br}$. In eq. (\ref{eq:dynSplit}) subscript $h$ denotes the front dynamic block, and subscript $s$ denotes the back dynamic block. $\gamma$ denotes a gain factor that depends on the particular pole and zero partition.
			
		\subsubsection{Estimating the static nonlinearity} \label{sec:nlEst}
			The static nonlinearity is estimated for every possible pole-zero partition $\{p_i,z_{j}^{i_{br}}\}$ of every branch $i_{br}$.
			
			This estimation is linear in the parameters when the nonlinearity is expressed as a linear combination of nonlinear basis functions (such as multivariate polynomial basis functions, piecewise linear basis functions, or radial basis function networks with a fixed width and a fixed center): 
			\begin{align}
				\hat{r}_i(k) &= \sum_{i_w=1}^{n_w} \hat{w}_{i_w}^{[i]} g_{i_w}(\hat{x}_1(k), \ldots, \hat{x}_{n_{br}}(k)),
			\end{align}
			where $\hat{w}_{i_w}^{[i]}$ is the coefficient belonging to the $i_w$-th basis function $g_{i_w}$ for the $i$-th output $\hat{r}_i(k)$ of the MIMO static nonlinearity, $\hat{x}_{j}(k)$ is the $j$-th input of the MIMO static nonlinearity, and $n_w$ is the number of nonlinear basis functions that is selected by the user.
			
			First, the intermediate signals $\boldsymbol{\hat{x}}^{\{p_i,\boldsymbol{z_{j}}\}}$ for pole partition $\{p_i\}$ and every possible zero partition $\{z_{j}^{i_{br}}\}$ of every branch $i_{br}$ are obtained:
			\begin{align}
				\hat{x}^{\{p_i,z_{j}^{i_{br}}\}}_{i_{br}}(k) &= \frac{\hat{B}^{\{z_{j}^{i_{br}}\}}_{h}(q)}{\hat{A}^{\{p_i\}}_{h}(q)} u(k),	\\
				\boldsymbol{\hat{x}}^{\{p_i,\boldsymbol{z_{j}}\}}(k) &= \left[ \begin{array}{ccc} \hat{x}^{\{p_i,z_{j}^{1}\}}_1(k) & \ldots & \hat{x}^{\{p_i,z_{j}^{n_{br}}\}}_{n_{br}}(k) \end{array} \right]^T, \nonumber \\
				\boldsymbol{z_{j}} &= \left[ \begin{array}{cccc} z_{j}^{1} & z_{j}^{2} & \ldots & z_{j}^{n_{br}} \end{array} \right].
			\end{align}
			Next, the MIMO nonlinearity is estimated from the intermediate signals $\boldsymbol{\hat{x}}^{\{p_i,\boldsymbol{z_{j}}\}}$ generated through the output filters of all the branches $i_{br}$ to the measured output. A regressor matrix $\boldsymbol{K}^{\{p_i,\boldsymbol{z_{j}}\}}$ is constructed using a fixed, user selected set of nonlinear basis functions $g_1$ to $g_{n_w}$. For one partition of poles and zeros $\{p_i,\boldsymbol{z_{j}}\}$ one obtains:			
		\begin{align}
			& \boldsymbol{K}^{\{i_{br},p_i,\boldsymbol{z_{j}}\}} = \nonumber \\
			 	& \left[  \begin{array}{ccc} 		
				\frac{\hat{B}^{\{z_j^{i_{br}}\}}_{s}(q)}{\hat{A}^{\{p_i\}}_{s}(q)} g_1(\boldsymbol{\hat{x}}^{\{p_i,\boldsymbol{z_{j}}\}}(1)) & \ldots & \frac{\hat{B}^{\{z_j^{i_{br}}\}}_{s}(q)}{\hat{A}^{\{p_i\}}_{s}(q)} g_{n_w}(\boldsymbol{\hat{x}}^{\{p_i,\boldsymbol{z_{j}}\}}(1))	\\ 
				\vdots & \ddots & \vdots 	\\
				\frac{\hat{B}^{\{z_j^{i_{br}}\}}_{s}(q)}{\hat{A}^{\{p_i\}}_{s}(q)} g_1(\boldsymbol{\hat{x}}^{\{p_i,\boldsymbol{z_{j}}\}}(N)) & \ldots &\frac{\hat{B}^{\{z_j^{i_{br}}\}}_{s}(q)}{\hat{A}^{\{p_i\}}_{s}(q)} g_{n_w}(\boldsymbol{\hat{x}}^{\{p_i,\boldsymbol{z_{j}}\}}(N))
			 	\end{array} \right], \nonumber
		\end{align}
		\begin{align} \label{eq:SnlEstReg}
			\boldsymbol{K}^{\{p_i,\boldsymbol{z_{j}}\}} 	&= \left[  \begin{array}{ccc} \boldsymbol{K}^{\{1,p_i,\boldsymbol{z_{j}}\}} & \ldots & \boldsymbol{K}^{\{n_{br},p_i,\boldsymbol{z_{j}}\}} \end{array} \right],
		\end{align}
		where $N$ is the total number of data points used.
		
		The coefficients of the nonlinear basis functions for the partition $\{p_i,\boldsymbol{z_{j}}\}$ are obtained using a linear least squares estimation:
		\begin{align}
			\boldsymbol{\hat{w}}^{\{p_i,\boldsymbol{z_{j}}\}}	&= \left( {\boldsymbol{K}^{\{p_i,\boldsymbol{z_{j}}\}}}^T \boldsymbol{K}^{\{p_i,\boldsymbol{z_{j}}\}} \right)^{-1} \boldsymbol{K}^{\{p_i,\boldsymbol{z_{j}}\}} \boldsymbol{y}, \label{eq:nlEst} \\
			\boldsymbol{y} &= \left[ y(1) \: y(2) \: \ldots \: y(N) \right] ^T
		\end{align}
		In practice, the solution is obtained using a QR decomposition. To improve the numerical conditioning of the matrix, the columns of $\boldsymbol{K}^{\{p_i,\boldsymbol{z_{j}}\}}$ are normalized. Each column is therefore divided by its $l^2$-norm.		
			
		\subsubsection{Pole-zero pattern selection} \label{sec:modelSelection}
			The simulation error $\boldsymbol{\hat{e}}^{\{p_i,\boldsymbol{z_{j}}\}}$ present between the modeled output and the measured output is computed. The partition that results in the lowest root mean square error is selected. From this point on, the front and the back LTI blocks, $\hat{H}^{[i]}(q,\boldsymbol{\theta})$ and $\hat{S}^{[i]}(q,\boldsymbol{\theta})$, and the coefficients of the static nonlinearity $\hat{w}_{i_w}^{[i]}$ are all estimated.
		
		The modeled output $\hat{y}(k,\boldsymbol{\theta})$ is obtained as follows:
		\begin{align}
			\hat{x}_i(k,\boldsymbol{\theta}) &= \hat{H}^{[i]}(q,\boldsymbol{\theta}) u(k), \\
			\hat{r}_i(k,\boldsymbol{\theta}) &= \sum_{i_w=1}^{n_w} \hat{w}_{i_w}^{[i]} g_{i_w}(\hat{x}_1(k,\boldsymbol{\theta}), \ldots, \hat{x}_{n_{br}}(k,\boldsymbol{\theta})), \\
			\hat{y}(k,\boldsymbol{\theta}) &= \sum_{i=1}^{n_{br}} \hat{S}^{[i]}(q,\boldsymbol{\theta}) \hat{r}_i(k,\boldsymbol{\theta}), \label{eq:ParWhModelOuptut}
		\end{align}
		where the signals are as in Figure \ref{fig:parallelWienerHammersteinModel}. The parameters of the model are stored in the parameter vector $\boldsymbol{\theta}$.
		
	\subsection{Improving the estimated nonlinearity} \label{sec:reEstimateNl}
		The number of parameters used by a MIMO static nonlinear model that is linear in the parameters tends to grow very fast. It grows combinatorially in the case of a multivariate polynomial for an increasing number of inputs and outputs, and for an increasing model complexity (e.g. the degree of the multivariate polynomial). Static nonlinear models that are nonlinear in the parameters, such as neural networks, can be less sensitive to this problem if properly tuned. For a standard feed-forward neural network with one hidden layer and a linear output layer, the number of parameters grows linearly with the number of input and outputs, and linearly with the complexity (number of neurons) of the neural network.
		
		An initial estimate of the nonlinear behavior and the LTI blocks that are present in the parallel Wiener\-/Hammerstein model can be obtained using one set of nonlinear basis functions resulting in a model that is linear in the parameters, e.g. using multivariate polynomials. In a second step, the static nonlinearity can be re-estimated using another MIMO static nonlinear model, e.g. using a neural network, to increase the model flexibility without increasing the number of parameters too much. The decision whether or not to perform this refinement step is left to the user. This step is easily performed as the intermediate signals $\hat{x}_i$ and $\hat{r}_i$, defined in Figure \ref{fig:parallelWienerHammersteinModel}, can be obtained using the model estimated in Section \ref{sec:modelSelection}. The initial guess of the parameters of this second parameterization can then be further refined in a final complete optimization step, as described in Section \ref{sec:optimization}.
		
\section{Persistence of excitation} \label{sec:persistence}
	\begin{assumption} \label{ass:persis}
		The input signal $u(k)$ is assumed to be persistently exciting the system. 
	\end{assumption}
	
	The assumption that the excitation is persistent is a very common assumption in system identification. This section discusses what persistence of excitation means for the proposed identification procedure.
	
	The first step in the identification algorithm is to identify the parametric BLA of the nonlinear parallel Wiener-Hammerstein system. It is important that the BLA identifies the dynamics that are present in the system correctly. Therefore, the number of excited frequencies in the input signal $u(k)$ needs be equal or higher than $\frac{n_d + n_c +1}{2}$.
	
	Also the MIMO static nonlinearity needs to be estimated. For this identification to work, the matrix $\boldsymbol{K}^{\{p_i,\boldsymbol{z_{j}}\}}$ in eq. (\ref{eq:SnlEstReg}) needs to be of full rank. Put in other words, the nonlinear basis functions $g_{i_w}(\hat{x}_1(k), \ldots, \hat{x}_{n_{br}}(k))$ need to be linearly independent over the domain of the intermediate signals $\hat{x}_1(k), \ldots, \hat{x}_{n_{br}}(k)$. A consequence is that the range of amplitudes present in $\hat{x}_1(k), \ldots, \hat{x}_{n_{br}}(k)$ needs to be sufficiently large.
	
	Furthermore, Assumption \ref{ass:D} does not only have consequences for the system. It also determines the choice of the different setpoints of the input signals. The setpoints are chosen to ensure that the rank of the matrix $\boldsymbol{D}$ is equal to $n_{br}$.
	
\section{Consistency of the initial estimates} \label{sec:consistency}
	This section shows the consistency of the proposed estimator when a linear-in-the-parameters nonlinearity model is used to describe the MIMO static nonlinearity.
	
	\begin{assumption} \label{ass:modelSet}
		The data is generated by a parallel Wiener-Hammerstein system that lies in the model set. 
	\end{assumption}	
	
	\begin{thm} \label{theo:BlaParWH1}
		The parameterized BLA $\hat{G}_{bla}^{[i_r]}\left(q,\hat{\boldsymbol{\theta}}_{bla}\right)$ in eq. (\ref{eq:parBla}) is a consistent (convergence with probability 1) estimate of eq. (\ref{eq:BLA}) when the number of samples $N$ tends to infinity, and the number of input signal realizations $M \geq 4$ under Assumptions \ref{ass:input}, \ref{ass:noise}, \ref{ass:modelSet}.
	\end{thm}
	\begin{proof}
		See Section 10.7 and Theorems 10.3 and 9.21 in \cite{Pintelon2012} combined with Theorem \ref{theo:BLA}.
	\end{proof}
	Since a nonparametric noise model is used during the identification, a minimum of 4 realizations $M$ is required to obtain convergence of the parametric BLA estimate to its expected value (see Theorem 10.3 in \cite{Pintelon2012}). This can be relaxed if a parametric rather than a nonparametric noise model is estimated.
			
	\begin{thm} \label{theo:consistency}
		The proposed estimator is a consistent (with probability 1 for $N \rightarrow \infty$) estimator of the class of parallel Wiener\-/Hammerstein systems defined by Assumptions \ref{ass:system}, \ref{ass:evenNl}, \ref{ass:polezero}, and \ref{ass:poles} for the  Riemann equivalence class of asymptotically normally distributed excitation signals (Assumption \ref{ass:input}), under the standard assumption of zero-mean additive noise at the output only (Assumption \ref{ass:noise}), and the persistence of excitation condition (Assumption \ref{ass:persis}). Furthermore, the system should be contained in the reachable model set (Assumption \ref{ass:modelSet}) for the estimated parameters to converge to the true parameters of the system, up to the degenerations of the model.
	\end{thm}	
	\begin{proof}
			Due to Assumptions \ref{ass:evenNl}, \ref{ass:polezero}, and \ref{ass:persis} and Theorems \ref{theo:BLA} and \ref{theo:BlaParWH1}, the matrix $\hat{\boldsymbol{D}}$ defined in eq. (\ref{eq:D}) is of low rank. The rank of the matrix $\hat{\boldsymbol{D}}$ is a consistent estimate of the number of parallel branches that is present in the system. The columns of the matrix $\boldsymbol{V}_{bla}$ that correspond to the significant singular values are a consistent estimate for the numerators, hence the zeros that are present in each branch, up to the degeneration of the model structure that is explained in Section \ref{sec:identifiability}.
		
		In the last step of the estimation algorithm, the MIMO static nonlinearity is estimated (eq.~\eqref{eq:nlEst}) for every possible pole-zero allocation. This problem is linear in the parameters, and it is solved with a linear least squares approach. Under Assumption \ref{ass:poles}, the poles and zeros that are allocated in this step are consistent estimates of the true poles and zeros that are present in the system, up to the degenerations of the model structure, as discussed in the previous paragraphs. 
		
		The estimate of the static nonlinearity is consistent for the pole-zero allocation that corresponds to the pole-zero allocation of the true system under Assumption \ref{ass:system}. A bias error will be present for the other pole-zero allocations, since the selected pole-zero allocation does not correspond to the exact pole-zero allocation of the system. Thus, this step results in a consistent estimate of the LTI blocks and the static nonlinearity when considering the pole-zero allocation that results in the smallest estimation error. 
		
		The estimated parameters are consistent and converge to the true parameters under Assumption \ref{ass:modelSet} up to the degenerations of the model structure as explained in Section \ref{sec:identifiability}.
	\end{proof}
	
	\begin{remark}
		It has been observed that in practice the rank determination still works well for small values of $R$ (smaller than $n_d$, larger than $n_{br}$)and a finite number of samples $N$ and realizations $M$.
	\end{remark}

\section{Final optimization} \label{sec:optimization}
	Joining all the previous estimation steps allows one to obtain the model parameters as a succession of estimations of subsets of the parameter vector. Although this results in a consistent estimate when the number of data points $N$ tends to infinity, this typically yields a sub-optimal estimate for a finite number of data samples. To increase the efficiency of the estimator, one can fine-tune all the parameters simultaneously in a final nonlinear\-/in\-/the\-/parameters estimation step. The optimized parameters are obtained by calculating:
	\begin{align}
		\hat{\boldsymbol{\theta}} = \underset{\boldsymbol{\theta}}{\argmin} \: \sum_{k=1}^N \left( y(k)-\hat{y}(k,\boldsymbol{\theta}) \right)^{2},
	\end{align}
	where $\hat{y}(k,\boldsymbol{\theta})$ is the modeled output, depending on the parameters $\boldsymbol{\theta}$. Note that the parameter vector $\boldsymbol{\theta}$ contains all the parameters of the model.

	This cost function unfortunately is non-convex with respect to the parameters $\boldsymbol{\theta}$. A Levenberg-Marquardt algorithm \cite{Pintelon2012} is used to minimize the cost function in a numerically stable and reliable way. This algorithm converges to the local minimum of the cost function that is 'closest' to the initial parameter values. Hence, good initial values of the parameters are very important to ensure the good quality of the final estimates. The positive effect of the proposed estimation method is studied in Section \ref{sec:initStudy}.

\section{Computational aspects} \label{sec:computationAspects}

		The major part of the workload of the proposed estimation algorithm lies in the partitioning of the poles and zeros. Remember that all possible pole-zero partitions are tried in this step (Section \ref{sec:partitioning}). For each partition, a linear least squares estimation needs to be performed. This can be quite demanding with respect to the computation time. To be more specific, consider a BLA with $n$ poles and $n$ zeros. The number of combinations $n_{\mathit{comb}}$ that needs to be scanned is bounded by:
		\begin{align}
			2^{\frac{n}{2}} 2^{n_{br} \frac{n}{2}} \leq n_{\mathit{comb}} \leq 2^{n} 2^{n_{br} n},
		\end{align}
		where $n_{br}$ is the number of parallel branches of the model.
		
		The upper limit is reached when only real poles and zeros are present in the decomposition of the BLA, while the lower limit is reached when all poles and zeros of the BLA decomposition appear in complex conjugate pairs. Typically, most poles and zeros appear as complex conjugate pairs. In practical cases, the actual number of combinations to be scanned will therefore be closer to the lower limit.
		
		For example, consider a BLA of order $n_d = n_c = 10$ in both numerator and denominator, and a 2-branch model. This results in a maximum number of combinations equal to $2^{10}*2^{20}$, which is about one billion combinations. Fortunately, the minimum number is only $32768$. Scanning all possible combinations in the upper limit is clearly not feasible. Scanning all possible combinations for the lower limit of this example is possible, although it remains expensive.
		
		The number of combinations that needs to be scanned can be reduced further by making some extra assumptions or by including prior knowledge about the system. A common assumption is that the linear subsystems should be proper. This reduces the number of combinations to be scanned significantly:
		\begin{align}
			 \sum_{k=0}^{\frac{n}{2}} \left( \frac{\frac{n}{2}!}{k!(\frac{n}{2}-k)!} \right)^{n_{br}} \leq n_{\mathit{comb}} \leq \sum_{k=0}^{n} \left( \frac{n!}{k!(n-k)!} \right)^{n_{br}}
		\end{align}	
		Considering the same example as above, this results in maximum of $184756$, and minimum of $252$ combinations. Scanning all possible combinations in the upper limit is feasible in about a day (considering that trying one possibility takes about 0.5 seconds). Scanning all possible combinations of the lower limit of is fortunately done in a couple of minutes.
		
		The order of the separate LTI-blocks can be fixed at front, and this also reduces the number of combinations that need to be tested. Also, the speed of the algorithm can be improved further by using parallel computing techniques that are nowadays present in, for instance, Matlab and Mathematica.

\section{Measurement example} \label{sec:measurement}
A real-world measurement based identification is performed to illustrate the good performance of the proposed method. First, the measurement setup is introduced. Next, the different steps of the model estimation procedure are shown. Finally, the validation results are discussed.

	\subsection{Measurement setup}
		The device under test (DUT) is a 2-branch parallel Wiener\-/Hammerstein system. The front and back LTI blocks of each branch are third order continuous time IIR filters. The static nonlinearity of each branch is realized with a diode-resistor network.

		The rest of the measurement setup is similar to the setup described in \cite{SchoukensM2012b}. The signals are generated by an arbitrary waveform generator (AWG), the Agilent/HP E1445A, sampling at 625 kHz. An internal low-pass filter with a cut-off frequency of 250 kHz is used as a reconstruction filter for the input signal. The in- and output signals of the DUT are measured by the alias protected acquisition channels (Agilent/HP E1430A) sampling at 78 kHz. The AWG and acquisition cards are clocked by the AWG clock, and hence the acquisition is phase coherent to the AWG. Leakage errors are hereby easily avoided. Finally, buffers are added between the acquisition cards and the in- and output of the DUT to avoid that the measurement equipment would distort the measurements.
		
	\subsection{Input design}
			The generated input signal $u(k)$ is a random phase multisine \cite{Pintelon2012} containing $N=131072$ samples with a flat amplitude spectrum. The excited band ranges from $\frac{f_s}{N}$ to $f_{max} = 20$ kHz, viz.:
			\begin{align}
				u(k) &= A \sum_{n=1}^{n_{max}} \cos(2\pi n \frac{f_s}{N} k + \phi_n),
			\end{align}
			where $n_{max}$ is the integer number closest to $N\frac{f_s}{f_{max}}$. The phases $\phi_n$ are independent uniformly distributed random variables ranging from $\left[0,2\pi\right.\left[\right.$. Twenty independent random phase realizations of the multisines are used at each input level to determine the BLA using the robust method. The input signal is applied at 5 different rms values that are linearly distributed between 100 mV and 1 V.
			
			The signals are measured at a sampling frequency of 78 kHz, which is 8 times slower than the sampling frequency at the generator side. This results in measured input and output signals that contain $N=16384$ measured samples per period.
		
	\subsection{Model estimation}
		This section shows how the different steps of the estimation algorithm are applied on the measurement example. First, the BLA of the system is measured and parameterized. Next, the estimated dynamics are distributed over the different LTI blocks that are present in the model. Finally, the nonlinearity is estimated and a nonlinear optimization of all the parameters of the model is performed.
			
		\subsubsection{BLA estimation and parameterization}
			The BLA is estimated and parameterized as discussed in Sections \ref{sec:measBla} and \ref{sec:paramBla}. The BLAs are parameterized with a discrete time rational transfer function model, with a common denominator. The numerators and denominator are both of order 12. The FRFs of the parameterized BLAs are shown in Figure \ref{fig:blaMeas}. Figure \ref{fig:blaMeas} also shows the noise variance and the total variance on the estimated BLAs. The total variance is the variance that is generated by the nonlinear behavior of the system and by the noise that is present in the measurements \cite{Pintelon2012,Schoukens2012}. The small variation that can be observed in the shape of the FRF of the BLAs will prove to be sufficiently informative to decompose the dynamics over the parallel branches. 
			\begin{figure}
				\centering
					\includegraphics[width=0.95\columnwidth]{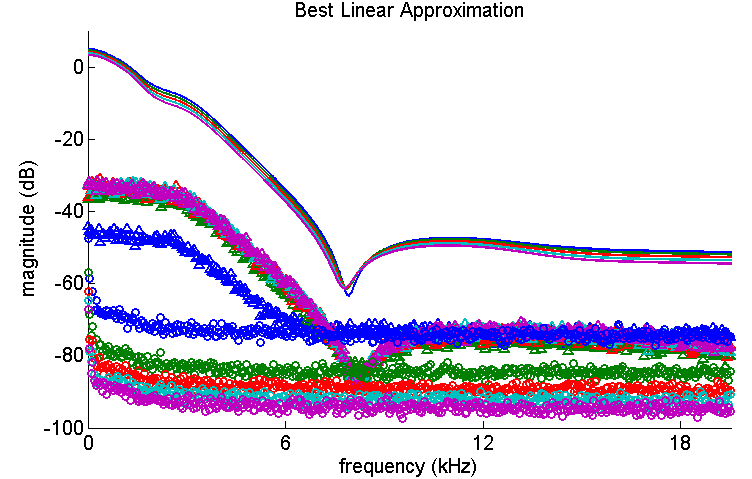}
				\caption{The parameterized BLA for the different excitation rms values. The FRF of the BLAs for the different rms values is shown by the full lines. The total variance on the BLAs is shown with the triangles. The noise variance on the BLAs is shown with the circles.}
				\label{fig:blaMeas}
			\end{figure}
		
		\subsubsection{Splitting the dynamics}
			The estimated dynamics are decomposed over the different parallel branches. Two parallel branches are retrieved by the SVD of the numerator matrix to model the system under test. The decomposed dynamics are then partitioned over the front and the back LTI blocks of the parallel Wiener\-/Hammerstein model. To do so, all the possible pole-zero combinations are scanned. It is assumed that all the LTI-blocks in the model are proper to reduce the number of possible combinations. As a result, a total of 140817 combinations are scanned. The mean square simulation error is used as an error criterion. The error is evaluated using one realization of both the lowest and the highest input excitation level of the estimation data.
			
			The lowest error after the pole-zero allocation scan is obtained with a model that has 4 poles and 4 zeros in the front LTI blocks, and 8 poles and 8 zeros in the back LTI blocks. This candidate model did not converge to a good local minimum after the final optimization step that is described in Section \ref{sec:optimization}. The second lowest error after the initial pole-zero scan (before the optimization step) is obtained with a model that has 6 poles and 6 zeros in the front LTI blocks, and 6 poles and 6 zeros in the back LTI blocks. This corresponds to the hardware realization of the system under test. This model is selected to be refined further in the next steps.
		
		\subsubsection{Estimating the static nonlinearity}
			A multivariate polynomial nonlinearity of order 7 is estimated during the partitioning of the dynamics to the front and the back LTI blocks. To increase the modeling power of the static nonlinear block, this polynomial nonlinearity is replaced by a 2-input 2-output neural network after the separation of the dynamics. The neural network has one hidden layer that contains 10 \textit{tanh(.)} activation functions, and a linear output layer. A \textit{tanh(.)} nonlinear function captures the saturation behavior in the system very well. Afterwards, a final simultaneous optimization of all the parameters is performed to further refine the estimated model.
		
	\subsection{Model validation}
		The estimated model is validated using two different signal types: random phase multisines of different magnitudes, and a growing envelope filtered Gaussian noise signal.
		
		\subsubsection{Multisine validation}
			The model is validated with a random phase multisine realization that is not used in the identification. The experiments are taken at 5 different rms values that are linearly distributed between 0.1 V and 1 V. The quality of the model is shown in Table \ref{tab:measValMulti} using three figures of merit: the rms value of the simulation error $\rms(e)$, the absolute mean value of the simulation error $\mu_e$, and the standard deviation of the simulation error $\sigma_e$, as defined below:
			\begin{align}
				\rms(e) &= \sqrt{\frac{1}{N}\sum_{k=1}^{N}e^2(k)},	\\
				\mu_e &= \left| \frac{1}{N}\sum_{k=1}^{N} e(k) \right|, \\
				\sigma_e &= \sqrt{\frac{1}{N-1}\sum_{k=1}^{N}(e(k)-\mu_e)^2},
			\end{align}
			where $e(k)$ is the difference between the measured output $y(k)$ and the simulated output $\hat{y}(k)$.

			\begin{table*}
				\center
				\caption{Validation error on a multisine signal} \label{tab:measValMulti}	
				\begin{tabular}{p{1.8cm} | p{0.6cm} p{0.6cm} p{0.6cm} | p{0.6cm} p{0.6cm} p{0.6cm} | p{0.6cm} p{0.6cm} p{0.6cm} | p{0.6cm} p{0.6cm} p{0.6cm} | p{0.6cm} p{0.6cm} p{0.6cm}}
					\multicolumn{16}{c}{Validation error (mV)} \\ \toprule[1.5pt]
					$\rms(u)$ 	& \multicolumn{3}{c|}{100}	& \multicolumn{3}{c|}{325}	& \multicolumn{3}{c|}{550}	& \multicolumn{3}{c|}{775}	& \multicolumn{3}{c}{1000} \\ \hline \hline
					Parallel WH & 0.30	& 0.30	& 0.02			& 0.50	& 0.32	& 0.38			& 0.38	& 0.38	& 0.03			& 0.57	& 0.57	& 0.08			& 1.10	& 1.06	& 0.30		 \\ \hline
					WH 					& 2.91 	& 2.90	& 0.31			& 7.36	& 7.25	& 1.32			& 10.43	& 10.41	& 0.60			& 15.11	& 15.08	& 1.02			& 20.24	& 20.20	& 1.28		 \\ \hline
					NARX 				& 3.20 	& 3.00	& 1.10			& 6.86	& 6.85	& 0.28			& 9.92	& 9.92	& 0.39			& 15.19	& 15.16	& 0.92			& 26.42	& 26.41	& 0.77		 \\ \hline
					NOE 				& 2.63 	& 2.62	& 0.24			& 4.90	& 4.89	& 0.13			& 4.44	& 4.44	& 0.03			& 5.54	& 5.54	& 0.10			& 18.55	& 18.54	& 0.59		 \\ \hline
					BLA 				& 1.34 	& 0.82	& 1.07			& 13.66	& 9.36	& 9.94			& 30.92	& 19.86	& 23.69			& 48.86	& 30.61	& 38.08			& 60.12	& 37.63	& 46.88	 	 \\ \bottomrule[1.25pt]
				\end {tabular}	
			\end{table*}
		
			The obtained model outperforms the BLA for every rms value of the input, as can be seen from Table \ref{tab:measValMulti}. Note that a different BLA is used for every rms value of the input, while only one parallel Wiener\-/Hammerstein model is used for all the different rms values of the input. The rms error is a combination of the standard deviation of the simulation error, and the mean value of the simulation error. The BLA is a linear approximation of the system, and cannot model the nonlinearities that are present in the system. The BLA can therefore not model the rms dependent constant contribution to the output that is generated by the nonlinearities. This explains the much larger mean error $\mu_e$ of the model output obtained with the BLA. Also the varying nonlinear contributions in the output cannot be explained by a linear model, and will contribute to the standard deviation of the simulation error. This explains the higher standard deviation of the simulation error. The parallel Wiener\-/Hammerstein model approximates the static nonlinearities that are present in the system quite well. Figure \ref{fig:measValMulti} shows that, indeed, the error on the modeled output of the BLA is coinciding with the level of the total variance on the measured output. This total variance is a measure for the nonlinear behavior of the system \cite{Pintelon2012,Schoukens2012}.
			
			The parallel Wiener\-/Hammerstein model output is compared with the results obtained by a Wiener\-/Hammerstein model in Table \ref{tab:measValMulti}. This Wiener\-/Hammerstein model is estimated similarly to the parallel Wiener\-/Hammerstein model, and uses a neural network with one hidden layer that contains 10 \textit{tanh(.)} activation functions and a linear output layer as a static nonlinearity. The Wiener\-/Hammerstein model is able to obtain a model error that is lower than the BLAs at the different excitation levels, but the errors are still 10 to 20 times larger than the errors of the parallel Wiener\-/Hammerstein model.
			
			The parallel Wiener\-/Hammerstein model is also compared with a neural network NARX model in Table \ref{tab:measValMulti}. The NARX input-output relationship is given by \cite{Billings2013}:
			\begin{align}
				y(k) &= f\left( u(k),\ldots,u(k-n_b),y(k-1),\ldots,y(k-n_a) \right) \nonumber \\
						 & + e(k),
			\end{align}
			where $n_b, n_a = 12$, $f(.)$ is a static nonlinear function, and $e(k)$ is white additive noise. Here, $f(.)$ is described by a neural network with one hidden layer that contains 25 \textit{tanh(.)} activation functions and a linear output layer. The estimation of the NARX model is performed using the Matlab Neural Network Toolbox using the so-called series-parallel architecture. The NARX model performs quite well, similar to the Wiener-Hammerstein model. The error obtained with the parallel Wiener-Hammerstein model is still 10 to 20 times smaller than the errors of the NARX model.
			
			The result that is obtained with the NARX model is further improved using a nonlinear output error model (NOE in Table \ref{tab:measValMulti}). Here, the delayed instances of the measured (noisy) outputs are no longer used in the regressor matrix, they are replaced by delayed instances of the noiseless output:
			\begin{align}
				\hat{y}(k) &= f\left( u(k),\ldots,u(k-n_b),\hat{y}(k-1),\ldots,\hat{y}(k-n_a) \right) \nonumber \\
							y(k) &= \hat{y}(k) + e(k),
			\end{align}
			where $\hat{y}$ denotes the noiseless output. This corresponds to the parallel architecture in the Matlab Neural Network Toolbox. The estimation of the parameters is performed using the Matlab Neural Network Toolbox. This results in an error which is over 30\% smaller than the error of the NARX model. However, the parallel Wiener-Hammerstein model still outperforms the NOE model (see Table \ref{tab:measValMulti}).
			
			The model error of the parallel Wiener\-/Hammerstein model is 30 to 40 dB lower than the total variance on the output (see Figure \ref{fig:measValMulti}), and it is only 10 dB higher than the output noise variance level. This shows that the proposed identification method captures the nonlinear behavior of the system very well. Therefore, it results in a high quality model.
			
			\begin{figure}
				\centering
					\includegraphics[width=0.95\columnwidth]{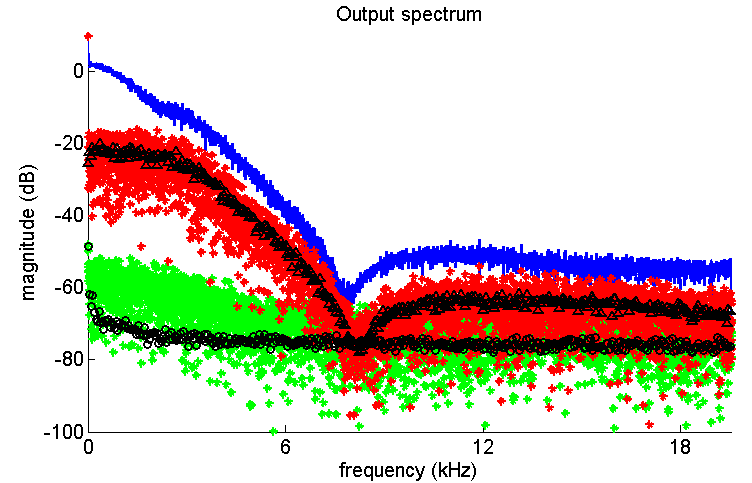}
				\caption{Measured and simulated output spectrum of a validation dataset. The measured output is shown in blue. The model error of the parallel Wiener\-/Hammerstein model is shown with the green plus symbols. The model error of the BLA is shown with the red plus symbols. The noise level at the system output is shown with the bottom black circles. The total distortion level at the output is shown with the top black triangles.}
				\label{fig:measValMulti}
			\end{figure}
		
		\subsubsection{Growing envelope validation}
			A second validation signal is used to assess the model quality over a broad amplitude range of the input in one signal. The input is a filtered Gaussian noise signal with an envelope that grows linearly over time:
			\begin{align}
				u(k) &= \frac{2k}{N}[H(q)r(k)],
			\end{align}
			where $r(k)$ is zero-mean white Gaussian noise with a standard deviation equal to one, and $H(q)$ is a 6th order low-pass Chebychev filter with a cut-off frequency located at 20 kHz and a passband ripple of 0.5 dB. Note that this is a generalization of the input signals that are used during the estimation. During the last part of the growing envelope input signal, the excitation amplitude is higher than the magnitude of the signals used in the estimation of the model. The rms value of the last portion of the growing envelope input signal is 1.4 V, where the maximum rms value during the estimation step was 1 V. This shows that the obtained model is even capable of extrapolating, although it is not advisable to rely on this property.
			
			\begin{figure}
				\centering
					\includegraphics[width=0.95\columnwidth]{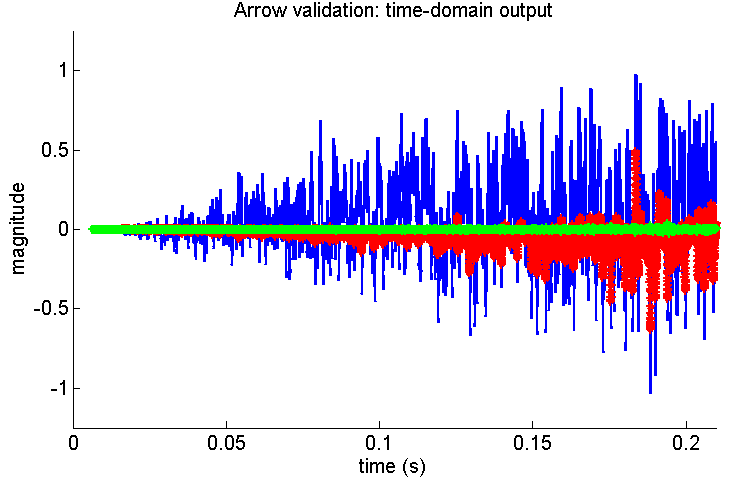}
				\caption{Time domain output of the growing envelope validation. The measured output is shown in blue. The model error of the parallel Wiener\-/Hammerstein model is shown with the green stars. The model error of the BLA is shown with the red stars.}
				\label{fig:arrowOutput}
			\end{figure}
			
			The parallel Wiener\-/Hammerstein model outperforms the BLA again. The results obtained for the different models are shown in Table \ref{tab:measValArrow} and in Figure \ref{fig:arrowOutput}. The BLA is obtained for an input rms value of 0.775 V. It is also clear from the obtained results that the model still performs well in the last quarter of the growing envelope input (after 0.15 seconds). This is the region where the model extrapolates. This proves the robustness of the obtained parallel Wiener\-/Hammerstein model with a neural network nonlinearity for this specific example. The rms errors of the BLA are about 10 to 20 times larger (20 to 26 dB) than the errors of the parallel Wiener\-/Hammerstein model. The Wiener\-/Hammerstein model, the NARX model and the NOE model are again able to obtain model errors that are lower than the model error of the BLA, but the errors are still about 5 to 10 times larger (20 dB) than the errors of the parallel Wiener\-/Hammerstein model.
			
			
			\begin{table*}
				\center
				\caption{Validation error on a growing envelope signal} \label{tab:measValArrow}	
				\begin{tabular}{p{1.8cm} | p{0.6cm} p{0.6cm} p{0.6cm} | p{0.6cm} p{0.6cm} p{0.6cm} | p{0.6cm} p{0.6cm} p{0.6cm} | p{0.6cm} p{0.6cm} p{0.6cm} | p{0.6cm} p{0.6cm} p{0.6cm}}
					\multicolumn{16}{c}{Validation error (mV)} \\ \toprule[1.5pt]
								& \multicolumn{3}{c|}{total} & \multicolumn{3}{c|}{quarter 1} & \multicolumn{3}{c|}{quarter 1} & \multicolumn{3}{c|}{quarter 1} & \multicolumn{3}{c}{quarter 1} \\
					$\rms(u)$ 	& \multicolumn{3}{c|}{822.30}	& \multicolumn{3}{c|}{179.07}	& \multicolumn{3}{c|}{522.89}	& \multicolumn{3}{c|}{889.42}	& \multicolumn{3}{c}{1268.1} \\ \hline \hline
					Parallel WH & 2.66	& 2.64	& 0.36			& 0.36	& 0.30	& 0.19			& 0.78	& 0.74	& 0.24			& 1.86	& 1.84	& 0.28			& 4.92	& 4.86	& 0.74		 \\ \hline
					WH 					& 20.20 & 20.20	& 0.03			& 4.44	& 4.36	& 0.82			& 10.22	& 10.21	& 0.47			& 17.32	& 17.29	& 1.02			& 34.77	& 34.77	& 0.15		 \\ \hline
					NARX 				& 18.77 & 18.53	& 3.01			& 3.52	& 3.00	& 1.83			& 8.60	& 8.60	& 0.20			& 17.76	& 17.54	& 2.80			& 31.75	& 30.93	& 7.22		 \\ \hline
					NOE 				& 22.93 & 22.32	& 0.82			& 1.98	& 1.80	& 0.82			& 6.28	& 5.81	& 2.39			& 17.50	& 15.45	& 8.21			& 41.89	& 40.76	& 9.69		 \\ \hline
					BLA 				& 55.74 & 46.70	& 30.44			& 11.19	& 11.10	& 1.49			& 31.08	& 22.42	& 21.53			& 55.78	& 33.12	& 44.89			& 90.70	& 72.99	& 53.86	 	 \\ \bottomrule[1.25pt]
				\end {tabular}	
			\end{table*}
			
		\subsection{Study of the initialization procedure} \label{sec:initStudy}
			A good initial estimate is a key factor to start the further optimization of the parameters if a high quality model is to be obtained. In this section we run the proposed algorithm until it arrives at the model selection step that is described in Section \ref{sec:modelSelection}. The models that correspond to the 100 best pole-zero allocations are optimized, and the models corresponding to 100 random pole-zero allocations are also optimized separately. All the pole-zero allocations that are considered have 6 poles and 6 zeros in the front LTI blocks and 6 poles and 6 zeros in the back LTI blocks to match with the system under test. The Levenberg-Marquardt optimization algorithm is stopped after 500 iterations, or sooner when convergence is reached. 
			
			It is clear from the results shown in Figure \ref{fig:randomVsBestInit} that the chance to obtain a good final model is higher when the best initial estimates are selected to be optimized further. The median error is more then 4 dB lower when the best initial estimates are selected (this is almost a factor 2 in rms error), compared with just picking randomly a pole-zero allocation set. Also, the variability of the final result is much lower when we start from the 100 best initial estimates.
			
			\begin{figure}
				\centering
					\includegraphics[width=0.95\columnwidth]{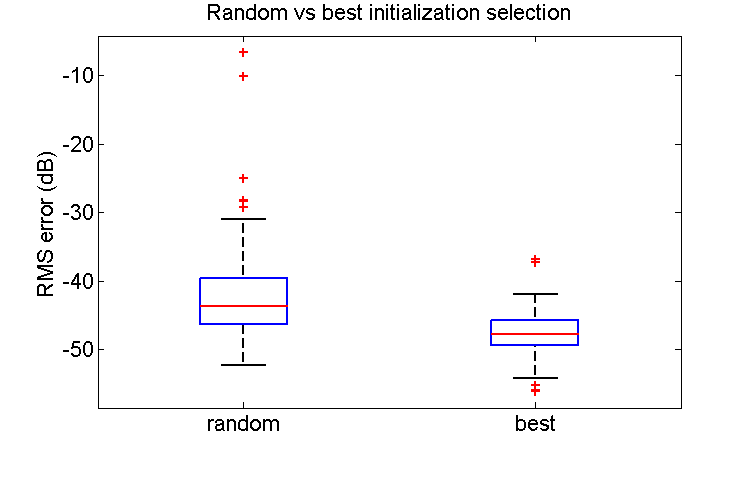}
				\caption{Boxplot of the rms error obtained after optimization using different initialization schemes. The boxplot on the left is obtained using a random pole-zero allocation over the LTI blocks of the model. The boxplot on the right uses the 100 pole-zero allocation resulting in the best candidate models.}
				\label{fig:randomVsBestInit}
			\end{figure}

\section{Conclusion} \label{sec:conclusion}
		An identification method for parallel Wiener\-/Hammerstein systems starting from input-output data only is presented. In the first step, the best linear approximation is estimated for different input excitation levels. In the second step, the dynamics are decomposed over a number of parallel orthogonal branches. Next, the dynamics of each branch are partitioned into a linear time-invariant subsystem at the input and a linear time-invariant subsystem at the output of each branch of the model. The static nonlinear block is also estimated during this step using a model that is linear in the parameters. This linear\-/in\-/the\-/parameters model can be replaced afterwards to increase the model flexibility. Finally, a nonlinear least squares optimization of the parameters of all blocks together is performed to refine the estimates. The consistency, and the computational complexity of the proposed initialization approach are discussed. The good performance of the proposed method, and the importance of a good initial estimate is illustrated on a measurement example. 
	 
\begin{ack}                               
	 This work was supported in part by the VUB (SRP-19), the Fund for Scientific Research (FWO-Vlaanderen), the Methusalem grant of the Flemish Government (METH-1), by the Belgian Government through the Inter university Poles of Attraction IAP VII/19 DYSCO program, and the ERC advanced grant SNLSID, under contract 320378. M. Schoukens is currently an FWO Aspirant, supported by FWO-Vlaanderen.
\end{ack}	

\bibliographystyle{plain}        													 
\bibliography{ReferencesLibraryV2}    										 

\begin{thebibliography}{10}

\bibitem{Bai1998}
E.W. Bai.
\newblock {An optimal two-stage identification algorithm for Hammerstein-Wiener
  nonlinear systems}.
\newblock {\em Automatica}, 34(3):333--338, 1998.

\bibitem{Baumgartner1975}
S.L. Baumgartner and W.J. Rugh.
\newblock Complete identification of a class of nonlinear systems from steady
  state frequency response.
\newblock {\em IEEE Trans. Circuits Syst.}, 22(9):753--759, 1975.

\bibitem{Billings2013}
S.A. Billings.
\newblock {\em Nonlinear System Identification: NARMAX Methods in the Time,
  Frequency, and Spatio-Temporal Domains}.
\newblock Wiley, West Sussex, UK, 1st edition, 2013.

\bibitem{Billings1978}
S.A. Billings and S.Y. Fakhouri.
\newblock Identification of a class of nonlinear systems using correlation
  analysis.
\newblock {\em Proc. IEE}, 125(7):691--697, 1978.

\bibitem{Billings1979}
S.A. Billings and S.Y. Fakhouri.
\newblock {Identification of non-linear Sm systems}.
\newblock {\em International Journal of Systems Science}, 10(10):1401--1408,
  1979.

\bibitem{Bussgang1952}
J.J. Bussgang.
\newblock {Cross-correlation functions of amplitude-distorted Gaussian
  signals}.
\newblock Technical Report 216, MIT Laboratory of Electronics, 1952.

\bibitem{Crama2004b}
P.~Crama and J.~Schoukens.
\newblock {Hammerstein-Wiener system estimator initialization}.
\newblock {\em Automatica}, {40}({9}):{1543--1550}, {2004}.

\bibitem{Enqvist2010}
M.~Enqvist.
\newblock {Identification of Block-oriented Systems Using the Invariance
  Property}.
\newblock In F.~Giri and E.W. Bai, editors, {\em Block-oriented Nonlinear
  System Identification}, volume 404 of {\em Lecture Notes in Control and
  Information Sciences}, pages 147--158. Springer, Berlin Heidelberg, 2010.

\bibitem{Enqvist2005a}
M.~Enqvist and L.~Ljung.
\newblock {Linear approximations of nonlinear FIR systems for separable input
  processes}.
\newblock {\em Automatica}, 41(3):459--473, 2005.

\bibitem{Gallman1975}
P.~Gallman.
\newblock {Iterative method for identification of nonlinear-systems using a
  Uryson model}.
\newblock {\em IEEE Trans. Autom. Control}, 20(6):771--775, 1975.

\bibitem{Giri2010}
F.~Giri and E.W. Bai, editors.
\newblock {\em {Block-oriented Nonlinear System Identification}}, volume 404 of
  {\em Lecture Notes in Control and Information Sciences}.
\newblock Springer, Berlin Heidelberg, 2010.

\bibitem{Ljung1999}
L.~Ljung.
\newblock {\em {System Identification: Theory for the User (second edition)}}.
\newblock Prentice Hall, Upper Saddle River, New Jersey, 1999.

\bibitem{Lyzell2012a}
C.~Lyzell, M.~Andersen, and M.~Enqvist.
\newblock {A Convex Relaxation of a Dimension Reduction Problem Using the
  Nuclear Norm}.
\newblock In {\em 51st IEEE Conference on Decision and Control (CDC)}, pages
  2852--2857, Maui, Hawaii, USA, Dec. 2012.

\bibitem{Palm1978}
G.~Palm.
\newblock {On representation and approximation of nonlinear systems}.
\newblock {\em Biological Cybernetics}, 31:119--124, 1978.

\bibitem{Palm1979}
G.~Palm.
\newblock {On representation and approximation of nonlinear systems Part II:
  Discrete Time}.
\newblock {\em Biological Cybernetics}, 34:49--52, 1979.

\bibitem{Pintelon1998}
R.~Pintelon, P.~Guillaume, G.~Vandersteen, and Y.~Rolain.
\newblock {Analyses, development and applications of TLS algorithms in
  frequency-Domain System Identification}.
\newblock {\em SIAM J. Matrix Anal. Appl}, 19(4):983--1004, 1998.

\bibitem{Pintelon2012}
R.~Pintelon and J.~Schoukens.
\newblock {\em {System Identification: A Frequency Domain Approach}}.
\newblock Wiley-IEEE Press, Hoboken, New Jersey, 2nd edition, 2012.

\bibitem{Pintelon2011b}
R.~Pintelon, G.~Vandersteen, J.~Schoukens, and Y.~Rolain.
\newblock {Improved (non-)parametric identification of dynamic systems excited
  by periodic signals-The multivariate case}.
\newblock {\em {Mechanical Systems and Signal Processing}},
  {25}({8}):{2892--2922}, {2011}.

\bibitem{Rolain1997}
Y.~Rolain, J.~Schoukens, and R.~Pintelon.
\newblock {Order Estimation for Linear Time-Invariant Systems Using Frequency
  Domain Identification Methods}.
\newblock {\em IEEE Trans. Autom. Contr.}, 42(10):1408--1417, 1997.

\bibitem{Schoukens2012}
J.~Schoukens, R.~Pintelon, and Y.~Rolain.
\newblock {\em {Mastering System Identification in 100 Exercises}}.
\newblock John Wiley \& Sons, Hoboken, New Jersey, 2012.

\bibitem{SchoukensM2012}
M.~Schoukens, E.W. Bai, and Y.~Rolain.
\newblock {Identification of Hammerstein-Wiener Systems}.
\newblock In {\em 16th IFAC Symposium on system identification}, pages
  274--279, Brussels, Belgium, Jul. 2012.

\bibitem{SchoukensM2013a}
M.~Schoukens, C.~Lyzell, and M.~Enqvist.
\newblock Combining the best linear approximation and dimension reduction to
  identify the linear blocks of parallel wiener systems.
\newblock In {\em 11th IFAC International Workshop on Adaptation and Learning
  in Control and Signal Processing (ALCOSP)}, pages 372--377, Caen, France,
  Jul. 2013.

\bibitem{SchoukensM2011}
M.~Schoukens, R.~Pintelon, and Y.~Rolain.
\newblock {Parametric Identification of Parallel Hammerstein Systems}.
\newblock {\em IEEE Trans. Instrum. Meas.}, 60(12):3931--3938, 2011.

\bibitem{SchoukensM2014}
M.~Schoukens, R.~Pintelon, and Y.~Rolain.
\newblock {Identification of Wiener-Hammerstein systems by a nonparametric
  separation of the best linear approximation}.
\newblock {\em Automatica}, 50(2):628--634, 2014.

\bibitem{SchoukensM2012b}
M.~Schoukens and Y.~Rolain.
\newblock {Parametric Identification of Parallel Wiener Systems}.
\newblock {\em IEEE Trans. Instrum. Meas.}, 61(10):2825--2832, 2012.

\bibitem{SchoukensM2014a}
M.~Schoukens, K.~Tiels, L.~Ishteva, and J.~Schoukens.
\newblock Identification of parallel wiener-hammerstein systems with a
  decoupled static nonlinearity.
\newblock In {\em 19th World Congress of the International Federation of
  Automatic Control}, pages 505--510, Cape Town, South Africa, Aug. 2014.

\bibitem{SchoukensM2013b}
M.~Schoukens, G.~Vandersteen, and Y.~Rolain.
\newblock {An identification algorithm for parallel Wiener-Hammerstein
  systems}.
\newblock In {\em 52nd IEEE Conference on Decision and Control (CDC)}, pages
  4907--4912, Florence, Italy, Dec. 2013.

\bibitem{Sjoberg2012a}
J.~Sj\"oberg, L.~Lauwers, and J.~Schoukens.
\newblock {Identification of Wiener-Hammerstein models: Two algorithms based on
  the best split of a linear model applied to the SYSID'09 benchmark problem}.
\newblock {\em Control Engineering Practice}, 20(11):1119--1125, 2012.

\bibitem{Sjoberg2012}
J.~Sj\"oberg and J.~Schoukens.
\newblock {Initializing Wiener-Hammerstein models based on partitioning of the
  best linear approximation}.
\newblock {\em Automatica}, 48(2):353--359, 2012.

\bibitem{Tiels2013}
K.~Tiels and J.~Schoukens.
\newblock {From coupled to decoupled polynomial representations in parallel
  Wiener-Hammerstein models}.
\newblock In {\em 52nd IEEE Conference on Decision and Control (CDC)}, pages
  4937--4942, Florence, Italy, Dec. 2013.

\bibitem{Vandersteen1997}
G.~Vandersteen, Y.~Rolain, and J.~Schoukens.
\newblock {Non-parametric Estimation of the Frequency-response Functions of the
  Linear Blocks of a Wiener-Hammerstein Model}.
\newblock {\em Automatica}, 33(7):1351--1355, 1997.

\bibitem{Westwick2012}
D.T. Westwick and J.~Schoukens.
\newblock {Initial estimates of the linear subsystems of Wiener-Hammerstein
  models}.
\newblock {\em Automatica}, 48(1):2931--2936, 2012.

\bibitem{Wysocki1976}
E.M. Wysocki and W.J. Rugh.
\newblock {Further results on the identification problem for the class of
  nonlinear systems Sm}.
\newblock {\em IEEE Trans. Circuits Syst.}, 23(11):664--670, 1976.

\end{thebibliography}

\end{document}